\documentclass[11pt]{article}
\usepackage{amsmath,amssymb,amsthm,booktabs,array}
\usepackage[margin=1in]{geometry}
\usepackage[unicode,colorlinks=true,linkcolor=blue,urlcolor=blue,citecolor=blue]{hyperref}
\newcommand{\cc}{\operatorname{cc}}
\newcommand{\OPT}{\operatorname{OPT}}
\newcommand{\Wone}{\mathrm{W}[1]}

\newcommand{\col}{\operatorname{col}}
\newcommand{\dcol}{\delta_{\mathrm{col}}}

\newtheorem{theorem}{Theorem}[section]
\newtheorem{lemma}[theorem]{Lemma}
\newtheorem{corollary}[theorem]{Corollary}
\theoremstyle{definition}

\newtheorem{definition}[theorem]{Definition}
\title{Minimum Label $k$-Cut on Planar Graphs:
Complexity and Color Structure}
\author{\small Yuanhao Wang$\,^{\rm a}$\quad\quad Wei Wang$\,^{\rm a}$\footnote{Corresponding author. Email address: wang\_weiw@163.com}\\
\small $^{\rm a}\,$School of Mathematics and Statistics, Xi'an Jiaotong University, Xi'an, 710049, P. R. China}
\date{}
\begin{document}
\maketitle
\begin{abstract}
Minimum Label $k$-Cut asks for a minimum-size set of colors such that deleting all edges of these colors leaves at least $k$ connected components. We study how the structure of the underlying graph and its color classes affects the complexity of this problem. Regarding fixed-parameter tractability with respect to $k$ on simple planar graphs, we prove that the problem is $\Wone$-hard, even when the underlying graph is a cactus graph and each color appears at most twice. Furthermore, for cactus graphs in which every color class is connected, we establish a complexity boundary in terms of the maximum monochromatic degree: when this degree is at most two, an optimal solution can be found in $O(Nk^2)$ time, where $N$ is the input size, with no restriction on color frequency or the number of blocks spanned by a color; allowing the maximum monochromatic degree to increase to three already makes the decision version NP-complete. Finally, we prove that the problem remains $\Wone$-hard parameterized by $k$ on fan graphs, even when every color class is an induced path. These results show that bounding color frequency does not suffice to eliminate parameterized hardness. They also establish a monochromatic-degree boundary for connected colors on cactus graphs and identify complexity barriers to extending this tractable case to outerplanar graphs.
\end{abstract}

\section{Introduction}\label{sec:introduction}

In various applications, edges sharing the same label must be handled as a group. In communication networks, edges subject to a shared failure risk can be represented by the same color~\cite{CPRV16}. Wang and Desmedt~\cite{WD11} used edge-colored graphs to model homogeneous faults: links employing the same hardware or software technology may fail simultaneously. Hou et al.~\cite{HWCZS15} used this model to design fault-tolerant heterogeneous networks and discussed its application to router interconnection design in networks-on-chip. Related directed label $s$-$t$ cut models also arise in security analysis using attack graphs~\cite{JSW02} and in sensor selection for supervisory control~\cite{RKK06}.

We study Minimum Label $k$-Cut on simple undirected edge-colored graphs. Each edge has exactly one color, and edges of the same color need not form a connected subgraph. The objective is to select as few colors as possible so that deleting all edges of these colors leaves at least $k$ connected components. All vertices are retained when edges are deleted.

\subsection{Related Work}
Previous work has mainly considered the global objective $k=2$ or label cuts with specified terminals. For $k=2$, Zhang~\cite{Z14} gave polynomial-time algorithms for planar graphs and graphs of bounded treewidth; Coudert et al.~\cite{CPRV16} studied the cases of connected colors and a bounded number of disconnected colors. Koehler et al.~\cite{Koehler} proposed mathematical programming formulations and exact solution algorithms. For general graphs, Jaffke et al.~\cite{JLMPS26} established a conditional quasi-polynomial lower bound and hardness with respect to the number of colors that are not deleted.

For a variable target number of components $k$, Fomin et al.~\cite{Hedge25} gave a randomized algorithm with running time $\binom{O(k\log n)+\ell}{\ell}n^{O(k)}p^{O(1)}$, where $\ell$ is the budget on the number of deleted colors, $n$ is the number of vertices, and $p$ is the number of colors. For every fixed $k$, this algorithm is FPT parameterized by $\ell$. However, the exponent of the input size still depends on $k$, so even an $O(k)$ upper bound on the color budget does not yield an FPT algorithm parameterized by $k$. The $k$-cut problem on general hypergraphs is $\Wone$-hard parameterized by $k+\ell$~\cite{CC22}, but this result does not directly apply to simple planar graphs. The algorithms for bounded color span in~\cite{CXY21,FPZ23} also have input-size exponents that depend on $k$, where the span of a color is the number of connected components of the subgraph formed by its edges. Fomin et al.~\cite{Hedge25} further asked whether Minimum Label $k$-Cut on simple planar graphs is fixed-parameter tractable parameterized by the target number of components $k$.

\subsection{Main Results}
Addressing this question, we prove that the problem is $\Wone$-hard parameterized by $k$, even when the underlying graph is a cactus graph and each color appears at most twice. For cactus graphs in which every color class is connected, if the maximum monochromatic degree is at most two, an optimal solution can be found in $O(Nk^2)$ time, where $N$ is the input size. The algorithm places no restriction on color frequency and allows a color to span arbitrarily many blocks. If the maximum monochromatic degree is allowed to increase to three, the decision version is already NP-complete, even when every color class is a single edge or a three-edge star.

If the underlying graph class is enlarged to outerplanar graphs, the problem remains $\Wone$-hard parameterized by $k$ even when every color class is a path; this hardness already holds on fan graphs. Both parameterized hardness results also hold for the color budget $\ell$ and the combined parameter $k+\ell$.

\section{Preliminaries}\label{sec:preliminaries}
Let $G=(V,E)$ be a simple undirected graph, and let $\col:E\to C$ be an edge-coloring map. We assume that every color in $C$ appears on some edge. For each color $c\in C$, let $E_c=\col^{-1}(c)$ be its color class and $f(c)=|E_c|$ its frequency. For a set of colors $L\subseteq C$, write $E(L)=\bigcup_{c\in L}E_c$. Deleting the color set $L$ means deleting all edges in $E(L)$ while retaining all vertices.

Let $\cc(J)$ denote the number of connected components of a graph $J$. Given a target number of components $k$, Minimum Label $k$-Cut asks to minimize $|L|$ subject to $\cc(G-E(L))\ge k$; its optimum value is denoted by $\OPT$. The decision version additionally specifies a color budget $\ell$ and asks whether there exists a color set satisfying this condition and $|L|\le\ell$.

Let $N=|V|+|E|+|C|$ be the input size and $c_0=\cc(G)$ the initial number of connected components. Let $h=k-c_0$ denote the number of additional components required, and let $\Psi(L)=\cc(G-E(L))-c_0$ denote the increase in the number of components caused by deleting the color set $L$. Thus, the graph has at least \(k\) connected components after deletion if and only if \(\Psi(L)\ge h\). If $k\le c_0$, no edge deletion is needed and the optimum value is zero; if $k>|V|$, no feasible solution exists. Henceforth, we consider only $c_0<k\le |V|$, that is, $0<h\le |V|-c_0$.

For each color $c$, let $G[c]$ be the subgraph consisting of its color class $E_c$ and the vertices incident with these edges. We call color $c$ connected if $G[c]$ is connected. The maximum monochromatic degree of the graph is defined as
\[
 \dcol=\max_{c\in C}\Delta(G[c]),
\]
that is, the maximum number of edges of the same color incident with a vertex; if all vertices of the graph are isolated, we define \(\delta_{\mathrm{col}}=0\).

\begin{definition}[Cactus Graph]
If every block containing an edge in a simple undirected graph $G$ is either a single-edge graph $K_2$ or a simple cycle, then $G$ is called a cactus graph. We allow cactus graphs to be disconnected and to contain isolated vertices.
\end{definition}

For a cactus graph $G$, let $\mathcal Q$ be its set of cycle blocks. For each $Q\in\mathcal Q$, let $C(Q)$ denote the set of colors appearing on the cycle block $Q\in\mathcal Q$. A parameterized problem is fixed-parameter tractable (FPT) if it can be solved in $g(\kappa)N^{O(1)}$ time, where $\kappa$ is the parameter, $g$ is a computable function, and the exponent of the input size is independent of $\kappa$. We consider three parameterizations: the target number of components $k$, the color budget $\ell$, and their sum $k+\ell$. Both of our parameterized hardness proofs are reductions from \textsc{Clique} parameterized by the clique size $t$. This problem is given a simple graph $H$ and an integer $t$ and asks whether $H$ contains a $t$-clique. In both reductions, the output parameters $k$ and $\ell$ are bounded by functions of $t$.

Using the notation above, the following two lemmas give, respectively, a formula for the increase in the number of components after edge deletion in a cactus graph and an upper bound on the optimal number of colors.

\begin{lemma}\label{H:lem:component}
If $G$ is a cactus graph, then for every edge set $F\subseteq E$,
\begin{equation}\label{H:eq:gain}
 \cc(G-F)-c_0=|F|-|\{Q\in\mathcal Q:F\cap E(Q)\ne\varnothing\}|.
\end{equation}
\end{lemma}
\begin{proof}
For any undirected graph $J$, its cycle rank is
\[
 \beta(J)=|E(J)|-|V(J)|+\cc(J).
\]
Every simple cycle of a cactus graph is a cycle block, and distinct cycle blocks have pairwise disjoint edge sets. Deleting one edge from each cycle block produces a forest with the same number of connected components as $G$.
The edge-count formula for forests gives $\beta(G)=|\mathcal Q|$. Let $q=|\{Q\in\mathcal Q:F\cap E(Q)\ne\varnothing\}|$, the number of cycle blocks from which at least one edge is deleted. Edge deletion creates no new cycles, and an original cycle block $Q$ remains a cycle after deletion if and only if $F\cap E(Q)=\varnothing$. Therefore, $G-F$ is still a cactus graph and has exactly $|\mathcal Q|-q$ cycle blocks. Hence $\beta(G-F)=|\mathcal Q|-q=\beta(G)-q$.

On the other hand, since all vertices are retained when edges are deleted, the cycle-rank formula gives
\[
\begin{aligned}
 &\beta(G-F)-\beta(G)\\
 =&\bigl(|E|-|F|-|V|+\cc(G-F)\bigr)
   -\bigl(|E|-|V|+c_0\bigr)\\
 =&-|F|+\cc(G-F)-c_0.
\end{aligned}
\]
Combining this with $\beta(G-F)=\beta(G)-q$, we obtain
\[
 \cc(G-F)-c_0=|F|-|\{Q\in\mathcal Q:F\cap E(Q)\ne\varnothing\}|.
\]
\end{proof}

A graph $G$ is $d$-degenerate if every nonempty subgraph of $G$
contains a vertex of degree at most $d$.

\begin{lemma}\label{D:lem:budget}
Let $G$ be a $d$-degenerate graph. If $c_0<k\le |V|$, then $\OPT\le dh$, where $c_0=\cc(G)$ and $h=k-c_0$.
\end{lemma}

\begin{proof}
Starting from $G$, delete edges iteratively until the number of connected components is at least $k$. Let $G'$ be the current graph. Since edge deletion preserves $d$-degeneracy, $G'$ is still $d$-degenerate.

If $\cc(G')<k$, then $k\le |V|$ implies that $G'$ has a connected component $H$ containing at least two vertices. By $d$-degeneracy, $H$ contains a vertex $v$ of degree at most $d$. Since $H$ is connected and has at least two vertices, $1\le d_H(v)\le d$. Deleting all edges incident with $v$ isolates $v$, and $H-v$ is nonempty, so the number of connected components increases by at least one.

Each step deletes at most $d$ edges, and at most $h$ steps suffice to reach the target. Let $F$ be the set of edges deleted in this process. Then $|F|\le dh$ and $\cc(G-F)\ge k$. Let $L=\{\col(e):e\in F\}$ be the set of colors used by these edges. Since $|L|\le |F|$ and $F\subseteq E(L)$, we have $\cc(G-E(L))\ge \cc(G-F)\ge k$. Thus, $L$ is a feasible solution containing at most $dh$ colors, and hence $\OPT\le dh$.
\end{proof}

Cactus graphs are $2$-degenerate, so $\OPT\le 2h$ whenever $c_0<k\le |V|$. Therefore, the dynamic program only needs to consider color budgets $0\le b\le \min\{|C|,2h\}$.

\section{Parameterized Hardness on Simple Planar Graphs}\label{sec:hardness}
The following construction is related to the selection principle in the Minimum $p$-Union problem~\cite{CDM17}. This problem asks for exactly $p$ sets from a given family whose union has minimum cardinality. Regard each edge of a simple graph $H$ as the set consisting of its two endpoints, and set $p=\binom t2$. Then there exist $p$ sets whose union has cardinality at most $t$ if and only if $H$ contains a $t$-clique. We realize this selection principle below by assigning colors to different cycle blocks.

\subsection{Reduction Construction}
Consider a \textsc{Clique} instance $(H,t)$, where $H$ is a simple undirected graph with vertices $v_1,\ldots,v_n$, and assume that $t\ge4$. Let $M=\binom t2$. If $n<t$ or $|E(H)|<M$, then $H$ cannot contain a $t$-clique. In this case, output a fixed no-instance: let $G$ be a single-edge graph, assign one color to its unique edge, and set $k=2$ and $\ell=0$. In what follows, assume that $n\ge t$ and $|E(H)|\ge M$, and construct a Minimum Label $k$-Cut instance $(G,\col,k,\ell)$.

For each source vertex $v$, create a cycle of length $q_v=\max\{3,d_H(v)\}$, denoted by $Q_v$, with all cycles pairwise vertex-disjoint. For each source edge $e=uv$, introduce a color $a_e$ and assign it to one as-yet uncolored edge in each of $Q_u$ and $Q_v$. Thus, each source-edge color appears exactly twice. Assign a new private color to each remaining cycle edge. Order the cycles arbitrarily and join consecutive cycles by bridges, each assigned a new private color. For each internal cycle, choose the two endpoints of a cycle edge as the attachment vertices of the bridges to the preceding and succeeding cycles, respectively. Set
\[
 \ell=M,\qquad k=2M-t+1=(t-1)^2.
\]
The private colors on the padding edges and bridges are collectively called auxiliary colors; each auxiliary color appears only once.

\subsection{Structural Properties and Correctness of the Reduction}
\begin{lemma}\label{H:lem:construction}
The graph $G$ obtained by the above construction is a connected simple cactus graph with maximum degree at most three and pathwidth at most two. Its edge-coloring is proper, and each color appears at most twice. The construction has size $O(n+|E(H)|)$ and can be carried out in polynomial time.
\end{lemma}
\begin{proof}
The cycles are pairwise vertex-disjoint and are joined only by a chain of bridges, so the resulting graph is a connected simple cactus graph. Each cycle vertex is incident with two cycle edges and at most one bridge, so the maximum degree is at most three. All colors within each cycle are distinct, and all bridge colors are new. The two occurrences of any source-edge color lie on different cycles, so the resulting edge-coloring is proper.

To give a path decomposition of width two, write the vertices of a cycle as $w_0,w_1,\ldots,w_{q-1}$, with the incoming bridge attached at $w_0$ and the outgoing bridge attached at $w_1$. Use the bags $\{w_0,w_{q-1},w_{q-2}\},\ldots,\{w_0,w_2,w_1\}$ in this order, followed by a bag of size two covering the bridge to the next cycle. Use the same bag sequence for the first and last cycles. The first cycle is joined to the next cycle only by the bridge at $w_1$, and the last cycle is joined to the preceding cycle only by the bridge at $w_0$. Every cycle edge is covered by a bag for its cycle, and every bridge is covered by its corresponding bridge bag. For each cycle, $w_0$ appears continuously from the incoming bridge bag to the last bag for the cycle, $w_1$ appears in the last bag for the cycle and the outgoing bridge bag, and every other vertex appears only in consecutive bags for that cycle. Hence, every vertex satisfies the consecutive-occurrence property. Each bag contains at most three vertices, so the pathwidth is at most two.

Since $q_v\le d_H(v)+3$, the total length of all cycles is at most $2|E(H)|+3n$, with another $n-1$ bridges between the cycles. The number of colors does not exceed the number of edges, so the output size is $O(n+|E(H)|)$. All the above steps can be carried out in polynomial time.
\end{proof}

\begin{lemma}\label{H:lem:equivalence}
The graph $H$ contains a $t$-clique if and only if the instance constructed above has a color set $L$ satisfying $|L|\le\ell$ and $\cc(G-E(L))\ge k$.
\end{lemma}
\begin{proof}

For the forward direction, given a $t$-clique, select the colors corresponding to its $M$ edges. Exactly $2M$ cycle edges are deleted, and exactly $t$ cycle blocks lose at least one edge. Since $G$ is connected, equation~\eqref{H:eq:gain} implies that the remaining graph has $1+2M-t=k$ connected components. The number of selected colors is $M=\ell$, so this is a feasible solution.

For the reverse direction, suppose that a feasible color set $L$ contains $a$ source-edge colors and $b$ auxiliary colors. Let $v$ be the number of cycle blocks from which at least one edge is deleted. Each source-edge color appears exactly twice, and each auxiliary color appears exactly once, so the number of deleted edges is $2a+b$. By equation~\eqref{H:eq:gain} and the budget constraint,
\[
 2a+b-v\ge2M-t,\qquad a+b\le M.
\]
Let the deficit of the number $a$ of selected source-edge colors from $M$ be $d=M-a\ge0$. The above inequalities give
\[
 b\le d,\qquad v+2d-b\le t,\qquad v+d\le t.
\]
Each selected source-edge color causes both of its corresponding cycle blocks to lose an edge. Thus, the $a$ selected source edges involve at most $v$ source vertices. Since $H$ is simple, $a\le\binom v2$. Suppose that $d\neq0$. Since $d\le t$,
\[
 M-d=a\le\binom v2\le\binom{t-d}{2}.
\]
However, for $1\le d\le t$ and $t\ge4$,
\[
 \binom t2-d-\binom{t-d}{2}
   =\frac{d(2t-d-3)}2>0.
\]
This is a contradiction. Therefore, $d=0$, and hence $a=M$, $b=0$, and $v\le t$. Moreover, $\binom t2=M=a\le\binom v2$ implies that $v=t$, and all $M$ possible edges among these vertices have been selected. Thus, they form a $t$-clique. The above argument also shows that no feasible color set contains an auxiliary color.
\end{proof}

\subsection{Complexity Consequences}
\begin{theorem}\label{H:thm:twice}
Minimum Label $k$-Cut, parameterized by $k$, $\ell$, or $k+\ell$, is $\Wone$-hard even on connected simple cactus graphs with maximum degree at most three and pathwidth at most two, under a proper edge-coloring in which each color appears at most twice.
\end{theorem}
\begin{proof}
\textsc{Clique} parameterized by the clique size $t$ is $\Wone$-complete. Since the cases $t\le3$ can be solved directly in polynomial time, restricting the source problem to $t\ge4$ does not affect its parameterized hardness.

The above preprocessing and construction can be performed in polynomial time. By Lemmas~\ref{H:lem:construction} and~\ref{H:lem:equivalence}, the output instance satisfies all the structural restrictions in the theorem and is equivalent to the source instance. Every output instance satisfies $k\le(t-1)^2$ and $\ell\le\binom t2$. Thus, the above reduction is an FPT reduction for each of the parameters $k$, $\ell$, and $k+\ell$, proving the result.
\end{proof}

\begin{corollary}
Unless $\mathrm{FPT}=\Wone$, the problem on simple planar graphs admits no exact algorithm running in $f(k)N^{O(1)}$ time. When colors are allowed to be disconnected, the hardness persists even if both treewidth and pathwidth are at most two, the graph has maximum degree at most three, the color frequency is at most two, and the maximum monochromatic degree is at most one.
\end{corollary}

\section{Complexity Boundaries for Connected Colors on Cactus Graphs}
\label{sec:structure}
This section studies the relationship between maximum monochromatic degree and computational complexity on cactus graphs with connected color classes. When the maximum monochromatic degree is at most two, each color class is a path or a cycle. At any cutvertex, a color can be adjacent along a cycle to at most two other colors, so the graph describing these local relationships consists only of paths, cycles, and isolated vertices. Combining this observation with the block structure of cactus graphs, we prove that the color--cycle-block incidence graph has treewidth at most three, and use this to obtain an exact algorithm with running time $O(Nk^2)$.

When the maximum monochromatic degree is allowed to reach three, color classes may branch at cutvertices. We construct a cactus graph in which all triangles share one vertex, so that the relationships between different colors forming three-edge stars realize the adjacency relationships of the source 3-regular graph, and transform the number of internal edges in the source graph into the increase in the number of components after deleting colors. This proves that the decision problem is already NP-complete even when each color class is only a single edge or a three-edge star.

\begin{theorem}\label{I:thm:poly}
If every color class of a cactus graph is connected and $\dcol\le2$, then, for $0<h=k-c_0\le |V|-c_0$, the optimum value and an optimal color set for Minimum Label $k$-Cut can be found in $O(N(h+1)^2)$ time.
\end{theorem}

\subsection{Incidence Structure Between Colors and Cycle Blocks}

By~\eqref{H:eq:gain}, the contribution of the number of edges of each color can be computed separately, whereas the deduction for a cycle block depends on whether any color on it is selected. To describe these relationships, we define the color--cycle-block incidence graph $I$, which is a bipartite graph whose two vertex sets correspond to the color set $C$ and the set $\mathcal Q$ of cycle blocks. An edge is added between $c$ and $Q$ if and only if color $c\in C$ appears on cycle block $Q\in\mathcal Q$. If the same color appears more than once on a cycle, only one incidence edge is added. Since $\sum_{Q\in\mathcal Q}|C(Q)|\le |E|$, the size of $I$ is $O(N)$.

For a color set $L\subseteq C$, its neighbor set $N_I(L)$ in $I$ consists precisely of all cycle blocks containing edges of selected colors. By Lemma~\ref{H:lem:component},
\begin{equation}\label{I:eq:gain}
 \Psi(L)=\sum_{c\in L}f(c)-|N_I(L)|.
\end{equation}

\begin{lemma}\label{I:lem:width}
If every color class of a cactus graph $G$ is connected and $\dcol\le2$, then a tree decomposition of $I$ of width at most three with $O(N)$ bags can be constructed in $O(N)$ time.
\end{lemma}

\begin{proof}
To construct a tree decomposition of $I$, we need to cover the incidence edges between colors and cycle blocks and ensure that the bags containing the same color form a connected set. The block--cutvertex forest of a cactus graph records how its blocks are connected. We use this forest to organize the connections between bags, placing the corresponding color vertices and cycle-block vertices in each bag. Consider the block--cutvertex forest of $G$. Its nodes are of two types, corresponding to the blocks of $G$ that contain edges and to its cutvertices; when a cutvertex $v$ belongs to a block $B$, an edge is added between the corresponding nodes. For a cycle-block node $Q$, set $X_Q=C(Q)\cup\{Q\}$; for a bridge node $e$, set $X_e=\{\col(e)\}$; for a cutvertex node $v$, set $X_v=\{c\in C:d_{G[c]}(v)>0\}$. We next construct a local graph within each bag, give a tree decomposition of width at most three for each local graph, and then connect these decompositions along the block--cutvertex forest.

Since every color class is connected, the block nodes and cutvertex nodes containing that color form a subtree of the block--cutvertex forest. Each connected component outside a block is adjacent to at most one vertex of that block, so a simple path cannot leave the block and then return. This shows that each color is also connected within every block in which it appears.

Let a cutvertex $v$ belong to a cycle block $Q$, and let $e_1,e_2$ be the two edges of $Q$ incident to $v$. By the definitions of the bags, the colors of both edges belong to $X_Q\cap X_v$. Conversely, if $c\in X_Q\cap X_v$, connectivity of the color class implies that there is a simple path of color $c$ from $v$ to another vertex in $Q$. This path lies entirely in $Q$, so at least one of $e_1,e_2$ has color $c$. Thus $X_Q\cap X_v=\{\col(e_1),\col(e_2)\}$, and this intersection contains at most two colors. For a bridge node $e$ and an adjacent cutvertex node $v$, we have $X_e\cap X_v=\{\col(e)\}$.

For each bag $X$, construct a local graph on vertex set $X$, starting with the induced subgraph $I[X]$. For each edge of the block--cutvertex forest, let $X,Y$ be the bags corresponding to its endpoints. If $X\cap Y=\{c,d\}$, add the edge $cd$ to both local graphs. Since the intersection of adjacent bags contains at most two vertices, after these edges are added, each intersection forms a clique in both local graphs.

Let $R_v$ denote the local graph corresponding to a cutvertex $v$. Each edge $cd$ of this graph comes from a cycle block whose two edges at $v$ have colors $c,d$, respectively. Each edge of color $c$ incident to $v$ creates at most one neighbor of $c$ in $R_v$. Hence $d_{R_v}(c)\le d_{G[c]}(v)\le2$. Therefore, every connected component of $R_v$ is a path, a cycle, or an isolated vertex, and a tree decomposition of $R_v$ of width at most two can be constructed in linear time.

Next consider the local graph corresponding to a cycle block $Q$. Since the edges of each color on $Q$ are connected, they form a path on $Q$ or cover the entire cycle. Denote these colors by $c_1,\ldots,c_m$ in their order around the cycle. In the local graph, $Q$ is adjacent to all these colors, and the added edges between colors join only pairs consecutive in the cyclic order. When $m\ge3$, the bags $\{Q,c_1,c_i,c_{i+1}\},\qquad 2\le i\le m-1$, taken in this order, cover all local edges, and each vertex occurs in consecutive bags; this decomposition has width at most three. When $m\le2$, take $C(Q)\cup\{Q\}$ as the only bag.

For a bridge node $e$, its local graph consists only of the color vertex $\col(e)$, so take $\{\col(e)\}$ as the only bag.

For every pair of adjacent nodes in the block--cutvertex forest, the intersection of their bags contains at most two vertices; if it contains two, these vertices are already joined in both local graphs. By the vertex and edge coverage requirements of a tree decomposition, each of the two local decompositions has a bag containing the entire intersection. Connect these two bags. Since the local decompositions are connected according to the structure of the block--cutvertex forest, the resulting structure is still a forest.

The construction covers all vertices of $I$, and every edge $cQ$ is covered by the local decomposition for cycle block $Q$. For a color $c$, the nodes corresponding to the original bags containing $c$ form a subtree of the block--cutvertex forest. Within each local decomposition, the bags containing $c$ form a connected set; when local decompositions are connected along this subtree, both bags chosen for each connection also contain $c$. Thus, all bags containing $c$ remain connected after the decompositions are joined. A cycle-block vertex $Q$ appears only in its own local decomposition, where its bags are likewise connected. Finally, add an empty bag and connect it to one bag in each decomposition tree to obtain a tree decomposition of $I$.

The total number of incidences between blocks and colors is at most $|E|$, and the total number of incidences between cutvertices and colors is at most $2|E|$. The number of added local edges is also $O(N)$. Hence the total size of all local graphs and their tree decompositions is $O(N)$. While constructing each local tree decomposition, also record the bags used for joining it to the others. Each edge of the block--cutvertex forest corresponds to one joining operation. The entire joining process takes $O(N)$ time, so all connections can be completed in linear time.
\end{proof}

\subsection{An Exact Algorithm for Maximum Monochromatic Degree at Most Two}
By Lemma~\ref{D:lem:budget}, an optimal solution selects at most $2h$ colors. Let $B=\min\{|C|,2h\}$. For $0\le j\le B$, write
\[
 M_j=\max_{\substack{L\subseteq C\\|L|=j}}\Psi(L).
\]
Introduce a binary variable $x_c$ for each color and a binary variable $y_Q$ for each cycle block. By~\eqref{I:eq:gain},
\begin{equation}\label{I:eq:implication}
 M_j=
 \max\left\{
   \sum_{c\in C}f(c)x_c-\sum_{Q\in\mathcal Q}y_Q:\sum_{c\in C}x_c=j,\quad x_c\le y_Q\ \text{for every }cQ\in E(I)
 \right\}.
\end{equation}
Here, $x_c=1$ means that color $c$ is selected. If a selected color occurs on cycle $Q$, the constraints force $y_Q=1$; otherwise, setting $y_Q$ to zero improves the objective value. Thus, in an optimal solution, $y_Q$ indicates precisely whether any edge of cycle $Q$ is deleted.

\begin{proof}[Proof of Theorem~\ref{I:thm:poly}]
By Lemma~\ref{I:lem:width}, construct a tree decomposition of $I$ of width at most three with $O(N)$ bags, choose any bag as the root bag, and determine the parent and child bags accordingly. Assign each variable to the bag closest to the root that contains its corresponding vertex, and count its objective contribution and the corresponding color budget only at that bag. For each constraint $x_c\le y_Q$, choose any bag containing both $c$ and $Q$, and check the constraint only at that bag.

For a bag $X$, let $S_X$ be its intersection with its parent bag; if $X$ is the root bag, set $S_X=\varnothing$. By the properties of a tree decomposition, if a variable occurs both in the subtree rooted at $X$ and outside that subtree, it must belong to $S_X$. Thus, when combining the subtree with the rest of the decomposition, it suffices to ensure consistency of the values on $S_X$ and to add the numbers of colors counted in the different parts. Feasible solutions with the same boundary assignment and color count have the same conditions for compatibility with the rest of the decomposition, so only the one with the largest contribution needs to be retained. Define the state $T_X(\sigma,j)$ to be the maximum total contribution of all variables assigned to bags in this subtree, subject to the boundary variables taking the values specified by $\sigma$, all constraints assigned to the subtree being satisfied, and exactly $j$ selected colors being counted. Here, $\sigma$ is a binary assignment to $S_X$; a color variable $x_c$ contributes $f(c)x_c$, and a cycle-block variable $y_Q$ contributes $-y_Q$.

To compute $T_X(\sigma,j)$, enumerate the values of the variables in $X\setminus S_X$ and combine them with the boundary assignment $\sigma$ to obtain an assignment $\eta$ to bag $X$. Thus, $\eta|_{S_X}=\sigma$. If $\eta$ violates any constraint assigned to bag $X$, discard this assignment. This bag counts only the color budget and objective contribution of variables in $X\setminus S_X$. Denote these quantities by
\[
 b_X(\eta)
 =\sum_{c\in C\cap(X\setminus S_X)}\eta(c)
\]
and
\[
 w_X(\eta)
 =\sum_{c\in C\cap(X\setminus S_X)}f(c)\eta(c)
  -\sum_{Q\in\mathcal Q\cap(X\setminus S_X)}\eta(Q).
\]

Let $\mathcal K_X$ be the set of child bags of $X$. For each child bag $Y$, the values on its boundary $S_Y=X\cap Y$ are determined by $\eta$, so we use the child state $T_Y(\eta|_{S_Y},j_Y)$. Here, $\eta|_{S_Y}$ denotes the restriction of $\eta$ to $S_Y$, and $j_Y$ is the color count of this child state. The sum of the color counts contributed by this bag and the child subtrees must equal $j$, and their objective contributions can be added. This gives the recurrence
\begin{equation}\label{I:eq:dp}
 T_X(\sigma,j)=
 \max_{\substack{
\eta|_{S_X}=\sigma,\ \eta\text{ satisfies the constraints assigned to }X\\
   j_Y\ge0,\quad b_X(\eta)+\sum_{Y\in\mathcal K_X}j_Y=j}}\left(w_X(\eta)+\sum_{Y\in\mathcal K_X}T_Y(\eta|_{S_Y},j_Y)\right).
\end{equation}
If $X$ is a leaf bag, there are no child states, and it suffices to maximize $w_X(\eta)$ over feasible assignments satisfying $b_X(\eta)=j$. If no assignment satisfies the conditions, set $T_X(\sigma,j)=-\infty$; otherwise, $T_X(\sigma,j)$ is the maximum contribution attainable by a feasible assignment, which may be negative.

The connectedness condition of the tree decomposition ensures that variables shared by different child subtrees occur in the current bag, so their values are consistently determined by $\eta$. The weight and color budget of each variable are counted only once, at the bag to which it is assigned, and each constraint is checked exactly once. Induction from the leaves toward the root of the tree decomposition then establishes the correctness of the states and recurrence. At the root bag $R$, we have $T_R(\varnothing,j)=M_j$.

The dynamic program needs to consider only color counts $0\le j\le B$. Each bag contains at most four variables, so at most $2^4$ binary assignments need to be enumerated. For each assignment satisfying the constraints of the bag, first count the colors and weights associated with the variables assigned to this bag, and then merge the dynamic programming tables of its child bags one at a time. Each merge only needs to enumerate the color counts of the part already merged and the current child bag; both counts are at most $B$, so the merge takes $O((B+1)^2)$ time. The tree decomposition has $O(N)$ bags, and each decomposition edge participates in only a constant number of merges. The total running time is therefore $O(N(B+1)^2)=O(N(h+1)^2)$. After obtaining the values $M_j$ at the root bag, return the smallest $j$ satisfying $M_j\ge h$, and recover the corresponding color set using the assignments and allocations of color counts stored in the transitions. For the decision version with budget $\ell$, it suffices to further truncate the range of color counts to $\min\{B,\ell\}$.
\end{proof}

\subsection{NP-Completeness for Maximum Monochromatic Degree at Most Three}

We now retain color connectivity and allow the maximum monochromatic degree to reach three. We first give a construction in which deleting the colors corresponding to a set of source vertices increases the number of connected components by exactly the number of edges induced by those vertices in the source graph.

Given a connected simple graph $H=(U,A)$ with at least two vertices, create a common vertex $o$. For each source edge $e=uv$, introduce two new vertices $x_e,y_e$ and construct a triangle $ox_ey_e$. Each source vertex $v$ corresponds to a color $a_v$,
and the edges of the triangle are colored as follows:
\[
 \col(ox_e)=a_u,\qquad
 \col(oy_e)=a_v,\qquad
 \col(x_ey_e)=z_e,
\]
where $z_e$ is a private color that occurs only once. Denote the resulting graph by $G$. It is a connected simple cactus graph whose triangles share only the vertex $o$. Each color class $a_v$ forms a star centered at $o$ with $d_H(v)$ edges.

\begin{lemma}\label{D:lem:windmill}
The construction satisfies the following equality for every integer $0\le t\le |U|$:
\begin{equation}\label{D:eq:windmill}
 \max_{\substack{L\subseteq C\\|L|\le t}}\Psi(L)=\max_{\substack{S\subseteq U\\|S|=t}}|E(H[S])|.
\end{equation}
\end{lemma}

\begin{proof}
For any color set $L$, write
\[
 S=\{v\in U:a_v\in L\},\qquad T=\{e\in A:z_e\in L\}.
\]
Then $|L|=|S|+|T|$.

For the triangle corresponding to a source edge $e=uv$, if its private color $z_e$ is not selected, it increases the number of components by one only when both $a_u,a_v$ are selected. If $z_e$ is selected, the increase contributed by this triangle equals the number of selected colors among $a_u,a_v$. Applying~\eqref{H:eq:gain} to each triangle gives
\begin{equation}\label{D:eq:private}
 \Psi(L)=|E(H[S])|+|\{e\in T:e\cap S\ne\varnothing\}|.
\end{equation}

Take any color set $L$ satisfying $|L|=|S|+|T|\le t$. By~\eqref{D:eq:private}, $\Psi(L)\le |E(H[S])|+|T|\le |E(H[S])|+t-|S|$. If $S\ne\varnothing$ and $|S|<t$, connectivity of $H$ implies that some vertex outside $S$ is adjacent to a vertex in $S$. Adding it to $S$ creates at least one additional internal edge. Repeating this operation extends $S$ to a set $S'$ of size exactly $t$ satisfying $|E(H[S'])|\ge |E(H[S])|+t-|S|\ge \Psi(L)$. If $|S|=t$, simply take $S'=S$. If $S=\varnothing$, then~\eqref{D:eq:private} gives $\Psi(L)=0$, and any vertex set of size $t$ satisfies the required inequality. Therefore, 
\[
\max_{\substack{L\subseteq C\\|L|\le t}}\Psi(L)\leq\max_{\substack{S\subseteq U\\|S|=t}}|E(H[S])|.
\]
Conversely, take any vertex set $S$ satisfying $|S|=t$, and let $L=\{a_v:v\in S\}$. Then $|L|=t$, and~\eqref{D:eq:private} gives $\Psi(L)=|E(H[S])|$. Hence 
\[
\max_{\substack{L\subseteq C\\|L|\le t}}\Psi(L)\geq\max_{\substack{S\subseteq U\\|S|=t}}|E(H[S])|.
\]
Combining the two directions proves~\eqref{D:eq:windmill}.
\end{proof}

\begin{theorem}\label{D:thm:np}
The decision version of Minimum Label $k$-Cut is NP-complete when the input is a connected simple cactus graph whose triangles share only one common vertex, every color class is connected, and the maximum monochromatic degree is at most three.
\end{theorem}

\begin{proof}
We reduce from \textsc{Minimum Bisection} on connected simple 3-regular graphs. Given a connected simple 3-regular graph $H$ with $2t$ vertices and a nonnegative integer $b$, this problem asks whether there is a set $S$ of size $t$ with at most $b$ edges between $S$ and its complement. NP-completeness under this restriction can be found in~\cite{JJK26}.

Since $|E(H)|=3t$, the source instance is necessarily a yes-instance when $b\ge3t$. In this case, output a triangle whose three edges have distinct colors, and set $k=2$, $\ell=2$. Deleting any two colors leaves exactly two connected components, so the output instance is also a yes-instance and satisfies the required graph and coloring restrictions. Henceforth, consider $0\le b<3t$.

Apply the triangle construction above to $H$, and set $\ell=t,k=1+\left\lceil\frac{3t-b}{2}\right\rceil$. Since $G$ is connected, the required increase in the number of components is $k-1=\left\lceil\frac{3t-b}{2}\right\rceil$. Let $\partial_H(S)$ be the set of source edges with exactly one endpoint in $S$. For any set with $|S|=t$, 3-regularity gives $3t=2|E(H[S])|+|\partial_H(S)|$. Thus,
\[
 |\partial_H(S)|\le b \quad\Longleftrightarrow\quad |E(H[S])| \ge \left\lceil\frac{3t-b}{2}\right\rceil.
\]
By Lemma~\ref{D:lem:windmill}, there is a vertex set $S$ of size $t$ satisfying the inequality on the right if and only if there are at most $t$ colors whose deletion leaves $G$ with at least $k$ connected components. This proves the equivalence of the reduction.

The triangles intersect only at vertex $o$, so $G$ is a connected simple cactus graph. Since $H$ is 3-regular, the subgraph of each color $a_v$ is a $K_{1,3}$, whereas each private color occurs only once. Thus every color class is connected and the maximum monochromatic degree is three. The construction has size $O(|U|+|A|)$ and can be completed in polynomial time.

Finally, given a color set, its size and the number of connected components after deleting the corresponding edges can be checked in polynomial time, so the decision problem belongs to NP.
\end{proof}

Theorems~\ref{I:thm:poly} and~\ref{D:thm:np} establish the monochromatic-degree boundary on cactus graphs with connected color classes: the problem is polynomial-time solvable when the maximum monochromatic degree is at most two; when the maximum monochromatic degree is allowed to reach three, the decision problem is already NP-complete even if each color occurs at most three times.

\section{Parameterized Hardness on Outerplanar Graphs with Path Color Classes}\label{sec:fans}
Theorem~\ref{I:thm:poly} shows that Minimum Label $k$-Cut can be solved in polynomial time when the underlying graph is a cactus and every color class is a path or a cycle. In this section, we retain the restriction that every color class is a path and investigate whether this tractability extends to outerplanar graphs. We prove that the problem remains $\Wone$-hard parameterized by $k$ even when the underlying graph is a fan graph. Thus, unless $\mathrm{FPT}=\Wone$, the preceding polynomial-time algorithm cannot be extended to all outerplanar graphs even under the restriction to path color classes.

For an integer $s\ge2$, take a path $P_s$ on $s$ vertices, add a new vertex $o$, and join it to every vertex of the path. The resulting graph is called a fan graph, where $o$ is called the center vertex and $P_s$ the rim path.

\begin{theorem}\label{F:thm:fan}
The decision version of Minimum Label $k$-Cut is $\Wone$-hard parameterized by $k$, $\ell$, and $k+\ell$, respectively, even when the input is a simple fan graph and the edges of each color form an induced path in the underlying graph.
\end{theorem}

\subsection{The Reduction Construction on Fan Graphs}
We reduce from \textsc{Clique} parameterized by the clique size $t\ge4$, with a simple source graph $H=(U,A)$. If $|U|<t$, we output a triangle whose three edges have distinct colors and set $k=3$ and $\ell=0$; this is a fixed no-instance satisfying the input restrictions. In what follows, assume that $|U|\ge t$ and set
\begin{equation}
 M=\binom t2,\qquad \ell=M+t,\qquad k=2M+1,\qquad q=\ell+1.
 \label{F:eq:parameters}
\end{equation}
For each $v\in U$, order its incident source edges arbitrarily as $e_1,\ldots,e_{d_H(v)}$. Create vertices $x_{v,e_1},\ldots,x_{v,e_{d_H(v)}}$ and $q$ guard vertices $p_{v,1},\ldots,p_{v,q}$, and join them into a path $P_v$ in the following order:
\[
 p_{v,1},\;x_{v,e_1},\ldots,x_{v,e_{d_H(v)}},\;p_{v,2},\ldots,p_{v,q}.
\]
The paths $P_v$ are pairwise vertex-disjoint. Since $q\ge2$, the first and last vertices of each $P_v$ are guard vertices; even if $v$ is isolated in the source graph, $P_v$ still contains at least one edge. Order the source vertices arbitrarily and, between each pair of consecutive paths, add an edge joining the last vertex of the former to the first vertex of the latter, obtaining a path $P$. Then add a vertex $o$ and join it to every vertex of $P$; denote the resulting graph by $G$.

For each source edge $e=uv$, introduce a color $\alpha_e$ and color both $ox_{u,e}$ and $ox_{v,e}$ with $\alpha_e$. For each source vertex $v$, introduce a new color $\beta_v$ and color every edge of the path $P_v$ with $\beta_v$. Assign a distinct new private color to each remaining edge; these private colors are collectively called auxiliary colors.

\begin{lemma}[Properties of the construction]\label{F:lem:model}
The graph $G$ is a simple fan graph of pathwidth at most two, and the edges of each color form an induced path in $G$. The size of the construction is $O(|U|t^2+|A|)$, and both output parameters $k$ and $\ell$ depend only on $t$.
\end{lemma}

\begin{proof}
By construction, $G$ is a simple fan graph. Denote the vertices of $P$ in order by $w_1,\ldots,w_s$. The bag sequence $\{o,w_1,w_2\},\ldots,\{o,w_{s-1},w_s\}$ gives a path decomposition of width two.

For each source edge $e=uv$, the edges of color $\alpha_e$ form the path $x_{u,e}-o-x_{v,e}$. Since the first and last vertices of both $P_u$ and $P_v$ are guard vertices, the subgraph of $G$ induced by the vertex set $\{x_{u,e},o,x_{v,e}\}$ is exactly the path $x_{u,e}-o-x_{v,e}$. The edges of color $\beta_v$ form the path $P_v$. Since $P_v$ is a contiguous subpath of $P$ and does not contain $o$, the subgraph of $G$ induced by its vertices is exactly $P_v$. Each auxiliary color occurs on only one edge, which also forms an induced path. Thus, all color classes satisfy the required condition.

The number of vertices of the path $P$ is $s=2|A|+q|U|$. The graph $G$ has $s+1$ vertices and $2s-1$ edges, and the number of colors does not exceed the number of edges. Since $q=O(t^2)$, the size of the construction is $O(|U|t^2+|A|)$. Together with $t\le |U|$, this shows that the entire construction can be carried out in polynomial time. Finally, \eqref{F:eq:parameters} shows that both $k$ and $\ell$ depend only on $t$.
\end{proof}

\subsection{Upper Bounds on the Component Increase}
Let $L$ be any color set satisfying $|L|\le\ell$. Define
\[
 A_L=\{e\in A:\alpha_e\in L\},\qquad
 S_L=\{v\in U:\beta_v\in L\},
\]
and write $a=|A_L|$ and $b=|S_L|$, with $g$ denoting the number of selected auxiliary colors. Thus, $a+b+g\le \ell=M+t$.

To facilitate counting components, delete from $G$ the edges joining distinct subpaths $P_v$, obtaining an auxiliary graph $G_0$. The paths $P_v$ are then connected to one another only through the vertex $o$, and $G_0$ has the same vertex set as $G$ and remains connected. For any color set $L$, additionally deleting these connecting edges from $G-E(L)$ cannot decrease the number of connected components, so
\begin{equation}\label{F:eq:comparison}
 \cc(G-E(L))\le \cc(G_0-E(L)).
\end{equation}

\begin{lemma}\label{F:lem:gain}
For any color set $L$ satisfying $|L|\le\ell$, the component increase after deleting these colors satisfies
\begin{equation}\label{F:eq:bounds}
 \Psi(L)\le 2a+g,\qquad \Psi(L)\le a+\binom b2+g.
\end{equation}
\end{lemma}

\begin{proof}
Let $d_{A_L}(v)$ denote the number of edges in $A_L$ incident with $v$, and let $g_0$ be the number of edges $op_{v,i}$ such that $v\in S_L$ and $\col(op_{v,i})\in L$.

If $v\notin S_L$, then $\beta_v\notin L$, so every edge of the subpath $P_v$ is retained. The edges $op_{v,1},\ldots,op_{v,q}$ have pairwise distinct private colors, and $q=\ell+1$. Hence, when at most $\ell$ colors are selected, at least one of these edges is retained. Therefore, all vertices of $P_v$ remain connected to $o$.

If $v\in S_L$, then every edge of $P_v$ is deleted. Since the edges joining distinct subpaths have already been removed from $G_0$, each vertex of $P_v$ can remain in the component containing the center only through its edge to $o$; if this edge is deleted, the vertex becomes isolated. For a vertex $x_{v,e}$, the edge $ox_{v,e}$ is deleted if and only if $e\in A_L$, so there are $d_{A_L}(v)$ isolated vertices of this type. Summing over all $v\in S_L$ and including the $g_0$ isolated vertices created by deleting edges $op_{v,i}$ accounts for all new components. All remaining vertices lie in a single connected component containing $o$, and hence
\begin{equation}\label{F:eq:gain}
 \cc(G_0-E(L))-1=\sum_{v\in S_L}d_{A_L}(v)+g_0.
\end{equation}

Each edge in $A_L$ has at most two endpoints in $S_L$, so $\sum_{v\in S_L}d_{A_L}(v)\le2a$. On the other hand, counting each edge with at least one endpoint in $S_L$ once, and then counting each edge with both endpoints in $S_L$ one additional time, gives
\[
 \sum_{v\in S_L}d_{A_L}(v)=|\{e\in A_L:e\cap S_L\ne\varnothing\}|+|A_L\cap E(H[S_L])|\le a+\binom b2.
\]
The last inequality uses the simplicity of $H$ and $|S_L|=b$. Since the edges counted by $g_0$ have pairwise distinct auxiliary colors, we have $g_0\le g$. Combining \eqref{F:eq:comparison} and \eqref{F:eq:gain} yields \eqref{F:eq:bounds}.
\end{proof}

\subsection{Correctness of the Reduction}
\begin{proof}[Proof of Theorem~\ref{F:thm:fan}]
If $S\subseteq U$ induces a $t$-clique, let
\[
 L=\{\beta_v:v\in S\}\cup\{\alpha_e:e\in E(H[S])\}.
\]
Then $|L|=t+M=\ell$. For each $e=uv\in E(H[S])$, both edges $ox_{u,e}$ and $ox_{v,e}$ are deleted, as are all edges of the subpaths $P_u$ and $P_v$. Since $x_{u,e}$ and $x_{v,e}$ are internal vertices of their respective subpaths, the edges joining distinct subpaths are not incident with them, so both vertices become isolated. This yields $2M$ isolated vertices. All other edges incident with $o$ are retained, so the remaining vertices belong to the same connected component containing $o$. Hence, $G-E(L)$ has exactly $2M+1=k$ connected components.

Conversely, suppose that $|L|\le\ell$ and $\Psi(L)\ge2M$. The first inequality in \eqref{F:eq:bounds} and $a+b+g\le M+t$ give
\begin{equation}\label{F:eq:budget}
 2M\le2a+g\le2(M+t-b-g)+g=2M+2(t-b)-g.
\end{equation}
Since $g\ge0$, we obtain $b\le t$. Suppose that $b<t$. Then the second inequality in \eqref{F:eq:bounds} gives
\[
 2M\le\Psi(L)\le a+\binom b2+g\le M+t-b+\binom b2.
\]
However, since $M=\binom t2$,
\begin{equation}\label{F:eq:strict}
 2M-\left(M+t-b+\binom b2\right)=\frac{(t-b)(t+b-3)}2>0,\qquad t\ge4,\quad 0\le b<t,
\end{equation}
a contradiction. Therefore, $b=t$.

Substituting $b=t$ into \eqref{F:eq:budget} yields $g=0$. The budget constraint now gives $a\le M$, while $\Psi(L)\ge2M$ and $\Psi(L)\le2a$ give $a\ge M$, so $a=M$. Combining \eqref{F:eq:comparison} and \eqref{F:eq:gain}, we have
\[
 2M\le\Psi(L)\le\sum_{v\in S_L}d_{A_L}(v)\le2|A_L|=2M.
\]
Thus, all the inequalities above hold with equality. Each edge in $A_L$ contributes at most two to the sum, so equality requires both of its endpoints to belong to $S_L$. Hence, there are $M=\binom t2$ distinct source edges among the $t$ vertices of $S_L$. Since $H$ is simple, $S_L$ induces a $t$-clique.

Together with Lemma~\ref{F:lem:model}, the construction gives FPT reductions for the parameters $k$, $\ell$, and $k+\ell$, respectively, proving the theorem.
\end{proof}

\section{Conclusion}
We have proved that Minimum Label $k$-Cut on simple planar graphs is $\Wone$-hard parameterized by $k$, thereby resolving the fixed-parameter tractability question posed by Fomin et al.\ under the assumption that $\mathrm{FPT}\ne\Wone$. This hardness persists on cactus graphs of maximum degree at most three and pathwidth at most two, with each color occurring at most twice. Thus, sparsity of the underlying graph and low color frequency are both insufficient to guarantee fixed-parameter tractability with respect to the target number of components.

For cactus graphs with connected color classes, we have established a boundary for polynomial-time solvability between maximum monochromatic degrees two and three. When the maximum monochromatic degree is at most two, an optimal solution can be found in $O(Nk^2)$ time, and the algorithm allows color classes to contain arbitrarily many edges and span arbitrarily many blocks. When the maximum monochromatic degree is allowed to reach three, the decision problem is already NP-complete even if every color class is a single edge or a three-edge star. Together, these two results show that, within this graph class, allowing connected color classes to branch changes the polynomial-time solvability of the problem.

The reduction on fan graphs further preserves color connectivity and the bound of two on the maximum monochromatic degree, yet still yields $\Wone$-hardness parameterized by $k$. Thus, the block structure of cactus graphs is an important condition for the positive result above; once the class of underlying graphs is enlarged to outerplanar graphs, requiring color classes to be paths is no longer sufficient to obtain a corresponding algorithm. Both parameterized hardness results in this paper also apply to the color budget $\ell$ and the combined parameter $k+\ell$, respectively.

\end{document}